\documentclass[a4paper,11pt]{article}

\usepackage{fullpage}

\usepackage{times}
\usepackage{soul}
\usepackage{url}
\usepackage[hidelinks]{hyperref}
\usepackage[utf8]{inputenc}
\usepackage[small]{caption}
\usepackage{graphicx}
\usepackage{amsmath}
\usepackage{booktabs}
\usepackage{algorithm}
\usepackage{algorithmic}
\usepackage{amsthm}
\usepackage{amssymb}

\usepackage{color}

\newtheorem{theorem}{Theorem}[section]
\newtheorem{corollary}[theorem]{Corollary}
\newtheorem{proposition}[theorem]{Proposition}
\newtheorem{lemma}[theorem]{Lemma}
\newtheorem{conjecture}[theorem]{Conjecture}

\theoremstyle{definition}
\newtheorem{definition}[theorem]{Definition}

\usepackage{natbib}

\usepackage{authblk}

\title{\bf Almost Envy-Freeness in Group Resource Allocation}

\author[1]{Maria Kyropoulou}
\author[2]{Warut Suksompong}
\author[1]{Alexandros A. Voudouris}

\affil[1]{School of Computer Science and Electronic Engineering, University of Essex}
\affil[2]{Department of Computer Science, University of Oxford}

\date{}

\begin{document}

\maketitle

\begin{abstract}
We study the problem of fairly allocating indivisible goods between groups of agents using the recently introduced relaxations of envy-freeness. We consider the existence of fair allocations under different assumptions on the valuations of the agents. In particular, our results cover cases of arbitrary monotonic, responsive, and additive valuations, while for the case of binary valuations we fully characterize the cardinalities of two groups of agents for which a fair allocation can be guaranteed with respect to both envy-freeness up to one good (EF1) and envy-freeness up to any good (EFX). Moreover, we introduce a new model where the agents are not partitioned into groups in advance, but instead the partition can be chosen in conjunction with the allocation of the goods. In this model, we show that for agents with arbitrary monotonic valuations, there is always a partition of the agents into two groups of any given sizes along with an EF1 allocation of the goods. We also provide an extension of this result to any number of groups.
\end{abstract}

\section{Introduction}
Fairness is one of the primary desiderata in decision-making procedures involving multiple agents. For instance, machine learning researchers have recently studied how to design classification systems that do not discriminate based on sensitive attributes such as race or gender \citep{DworkHaPi12}. Another problem that is ubiquitous in every society is that of allocating resources among its members. The study of how the allocation can be done fairly, commonly known as \emph{fair division}, has found applications ranging from settling divorce disputes \citep{BramsTa96} to sharing apartment rent \citep{Su99,GalMaPr17}.

The vast majority of the fair division literature has focused on allocating resources among \emph{individual} agents. However, in many practical situations the resources need to be allocated among \emph{groups} of agents. The agents in each group share the same set of resources, but may have different preferences over them. For instance, the books allocated to a library can be enjoyed by all of its members, and it may be the case that some members prefer detective novels while others would rather read science fiction. Another example is the division of household goods between families; different members of a family may have contrasting opinions on the television or the sofa in their apartment.

The group aspect of fair division was introduced independently by \citet{SegalhaleviNi16} and \citet{ManurangsiSu17}. Segal-Halevi and Nitzan investigated the allocation of \emph{divisible} goods such as cake or land. In contrast, Manurangsi and Suksompong studied the group allocation of \emph{indivisible} goods like books and cars. Both of these works used the fairness notion of \emph{envy-freeness}---an agent is said to be \emph{envy-free} if she finds her group's share to be as least as good as the share of any other group. While envy-freeness cannot be guaranteed even when allocating indivisible goods among individuals (consider two agents who quarrel over a single valuable good), Manurangsi and Suksompong showed that if the agents' utilities for the goods are drawn at random, an envy-free allocation exists with high probability in the group setting as the number of agents and goods grows. \citet{SegalhaleviSu18} then introduced the concept of \emph{democratic fairness}, which aims to satisfy a certain fraction of the agents in each group. Among other fairness notions, they considered \emph{envy-freeness up to one good (EF1)}, which means that while some agent may envy another group under the given allocation, the envy can be eliminated by removing a single good from the other group’s share. Segal-Halevi and Suksompong showed that for two groups, there always exists an allocation that is EF1 for at least $1/2$ of the agents in each group, and the factor $1/2$ cannot be improved. We refer to \citep{Suksompong18-2} for an overview of the group fair division literature.

While the aforementioned works provide different methods for extending individual fair division to the group setting, in some situations it may be important that all agents receive a fairness guarantee with certainty regardless of their valuations. \citet{Suksompong18} showed the possibilities and limitations of using the \emph{maximin share} notion to guarantee every agent a fair share. In this work, we study the extent to which the recently introduced relaxations of envy-freeness, most notably EF1 and another notion called EFX, can be used for the same purpose. We show that while EF1 is surprisingly robust and can be guaranteed in a number of group settings, this is not the case for EFX. In addition, we introduce a new model in which the partition of the agents into groups is not fixed in advance, but instead can be chosen in conjunction with the allocation of the goods. This model captures settings where agents {(or a central authority)} can choose the group that they want to be part of, such as membership in a library or gym.

\subsection{Our Results}

With the exception of Section~\ref{sec:variable-many}, we assume that the goods are allocated between two groups of agents. While this may seem restrictive at first glance, we remark here that fair division between two \emph{individual} agents, which is much more restrictive, has received a significant amount of attention in the literature (e.g., \citep{BramsFi00,BramsKiKl12,BramsKiKl14,Aziz16,KilgourVe18}). Indeed, as we will see, the setting of two groups is quite rich and already allows for many interesting, non-trivial results.

In Section~\ref{sec:fixed-binary}, we assume that the agents in the two predetermined groups have binary valuations, i.e., each agent either desires each good or not. We characterize the cardinalities of the groups for which an EF1 or EFX allocation always exists.
Additionally, we consider a stronger variant of EFX introduced by \citet{PlautRo18}, which we refer to as EFX$_0$. We prove a very strong negative result for the group fairness setting,
implying that this fairness notion can only be guaranteed when both groups are singletons.

Next, in Section~\ref{sec:fixed-general}, we consider more general classes of valuations.
If one group is a singleton and the other group consists of two agents, we show that a \emph{balanced} EF1 allocation always exists provided that the agents are endowed with responsive valuations, a general class that contains the well-studied class of additive valuations. Balancedness means that the sizes of the two bundles differ by at most one. Moreover, we establish a surprising connection between our group fair division problem and a class of graphs known as \emph{generalized Kneser graphs}. We show that if a conjecture by \citet{JafariAl17} on the chromatic number of particular graphs from this class is true, it would imply that a balanced EF1 allocation exists whenever the two groups contain a total of at most five agents with arbitrary monotonic valuations. This bound would be tight due to our results in Section~\ref{sec:fixed-binary}.

Finally, in Section~\ref{sec:variable} we examine the newly introduced setting where we assume that the partition of the agents into groups is no longer fixed and can be chosen along with the allocation of the goods. Our results indicate that if a central authority or the agents themselves have the power to decide which group to join, then fair allocations are much easier to achieve. In particular, we show that for two groups of agents with arbitrary monotonic valuations, it is always possible to simultaneously obtain a balanced partition of the agents and a balanced EF1 allocation of the goods. In addition, for any given sizes of the two groups, there is a partition of the agents conforming to those sizes together with an EF1 allocation of the goods. We also present an extension of this result to any number of groups.

\subsection{Further Related Work}

The fairness notions EF1 and EFX were introduced by \citet{LiptonMaMo04} and \citet{CaragiannisKuMo16}, and studied in several papers over the last few years \citep{PlautRo18,AmanatidisBiMa18,BiswasBa18,BeiLuMa19,BiloCaFl19,OhPrSu19,AmanatidisBiFi20,BercziBeBo20,Suksompong20}. For individual fair division, it is known that an EF1 allocation is guaranteed to exist for any number of agents with monotonic valuations, while the question remains open for EFX.

A paper by \citet{GhodsiLaMo18} addressed the problem of \emph{rent division} for groups. In addition to determining the allocation of the rooms, the rent of the apartment must be divided among the agents. Like us, Ghodsi et al. also considered a model where the groups are not predetermined.

Another line of research has also considered group fairness in resource allocation but using a different kind of fairness notions than ours \citep{Berliant92,Husseinov11,TodoLiHu11,AleksandrovWa18,ConitzerFrSh19,AzizRe19}. In these papers, the resources are allocated to \emph{individual} agents, and the aim is to minimize the envy that arises between groups of these agents. In contrast, in our work the resources are allocated to \emph{groups} of agents and shared as public goods among the agents within each group.

\cite{BenabbouChEl19} studied a setting closer to ours, where the resources are allocated to groups.
However, these resources are then further assigned to individual agents in each group, so each agent does not derive full utility from all resources allocated to their group.
\citet{ChakrabortyIgSu20} considered individual fair division but allowed agents to have different \emph{weights}. Their setting can therefore capture group fair division, with weights corresponding to the group sizes. Nevertheless, in this interpretation, each group is only represented by a single preference, whereas our model allows differing preferences within the same group.

\citet{BiswasBa18} examined cardinality constraints in individual fair division, where the goods are categorized and there is a limit on the number of goods from each category that can be allocated to each agent. Our balancedness notion can be seen as a special case of these constraints.

The group resource allocation problem with variable groups is similar to coalition formation problems \citep{DrezeGr80,BogomolnaiaJa02} in the sense that there is flexibility in how the agents form groups. However, in coalition formation the utilities of an agent depend on how the agents are grouped and there are no goods involved, whereas in our setting these utilities depend on how the goods are distributed and not how the agents are grouped.

\section{Preliminaries}

Let $G=\{g_1,\ldots,g_m\}$ denote the set of goods, and $A$ the set of $n$ agents. A \emph{bundle} is a subset of $G$. The agents are partitioned into $k\geq 2$ groups. We assume in Sections~\ref{sec:fixed-binary} and \ref{sec:fixed-general} that this partition is fixed in advance, and in Section~\ref{sec:variable} that the partition is variable and can be chosen. Denote by $n_i$ the size of group~$i$ (so $n=\sum_{i=1}^k n_i$), and $a_{ij}$ the $j$th agent of group $i$. The agents in each group will be collectively allocated a subset of $G$; denote by $B_i$ the bundle allocated to group $i$ so that $B_i\cap B_j=\emptyset$ for any $i\neq j$ and $\cup_{i=1}^k B_i=G$. A partition of the agents is called \emph{balanced} if $|n_i-n_j|\leq 1$ for any $i,j$. Similarly, an allocation of the goods is called balanced if $||B_i|-|B_j||\leq 1$ for any $i,j$.

Each agent $a_{ij}$ has some non-negative valuation $u_{ij}(G')$ for each set of goods $G'\subseteq G$; for convenience we write $u_{ij}(g)$ instead of $u_{ij}(\{g\})$ for a good $g\in G$. Let $\textbf{u}_{ij}:=(u_{ij}(g_1),\dots,u_{ij}(g_m))$ be the valuation vector of agent $a_{ij}$ for individual goods. 
When the identity of the agent is not important, we will drop the indices $i,j$ from the valuation $u_{ij}$.
We assume that valuations are 
\begin{itemize}
\item \emph{monotonic}: $u(G_1)\leq u(G_2)$ for any $G_1\subseteq G_2\subseteq G$, and
\item \emph{normalized}: $u(\emptyset)=0$.
\end{itemize}

\begin{definition}
A valuation function $u$ is said to be 
\begin{itemize}
\item \emph{responsive} if $u(G'\cup\{g\})\geq u(G'\cup\{\overline{g}\})$ for any $G'\subseteq G$ and any goods $g,\overline{g}\not\in G'$ such that $u(g)\geq u(\overline{g})$,
\item \emph{additive} if $u(G')=\sum_{g\in G'}u(g)$ for any $G'\subseteq G$, and
\item \emph{binary} if it is additive and $u(g)\in\{0,1\}$ for all $g\in G$.
\end{itemize}
\end{definition}

Note that every additive valuation is responsive. Additive valuations are often assumed in recent fair division literature
\citep{CaragiannisKuMo16,AmanatidisBiMa18,BiswasBa18,ConitzerFrSh19}.
An \emph{instance} consists of agents, goods, and utility functions (and in the model of Sections~\ref{sec:fixed-binary} and \ref{sec:fixed-general}, the partition of agents into groups). In Section~\ref{sec:variable}, we simply denote the agents by $a_1,\dots,a_n$ and their valuations by $u_1,\dots,u_n$.

We are now ready to define the fairness notions that we consider in this paper.

\begin{definition}
An allocation $(B_1,\dots,B_n)$ is said to be
\begin{itemize}
    \item \emph{envy-free for agent $a_{ij}$} if $u_{ij}(B_i)\geq u_{ij}(B_{i'})$ for any $i'$;
    \item \emph{envy-free up to any good (EFX$_0$) for agent $a_{ij}$} if for any $i'$ and any good $g\in B_{i'}$, we have $u_{ij}(B_i)\geq u_{ij}(B_{i'}\backslash\{g\})$;
    \item \emph{envy-free up to any positively valued good (EFX) for agent $a_{ij}$} if for any $i'$ and any good $g\in B_{i'}$ such that $u_{ij}(g)>0$, we have $u_{ij}(B_i)\geq u_{ij}(B_{i'}\backslash\{g\})$;
    \item \emph{envy-free up to $c$ goods (EF$c$) for agent $a_{ij}$}, for a given positive integer $c$, if for any $i'$ there is a set $B\subseteq B_{i'}$ with $|B|\leq c$ such that $u_{ij}(B_i)\geq u_{ij}(B_{i'}\backslash B)$.
\end{itemize}
An allocation is said to be \emph{envy-free} if it is envy-free for every agent. When there are two groups, we say that an agent finds a bundle to be \emph{envy-free} if the allocation that assigns the bundle to her group and the complement bundle to the other group is envy-free for her. Analogous definitions hold for EFX$_0$, EFX, and EF$c$.
\end{definition}

EFX$_0$ is a variant of EFX introduced by \citet{PlautRo18}.\footnote{In their paper this property is simply called EFX; we rename it to avoid confusion with the original definition of \citet{CaragiannisKuMo16}, which we refer to as EFX.} For additive valuations, it is clear that each property in the list is stronger than the next, with EFX implying EF1. We will only consider EFX and EFX$_0$ in the context of additive valuations. In Sections~\ref{sec:fixed-general} and \ref{sec:variable} we only state results for EFX, but all of these results also hold for EFX$_0$.

\section{Fixed Groups with Binary Valuations}
\label{sec:fixed-binary}

In this section, we assume that the agents have binary valuations and are partitioned in advance into two groups of size $n_1$ and $n_2$. Note that any non-existence result for $(n_1,n_2)$ yields an analogous result for $(n_1',n_2')$ with $n_1'\geq n_1$ and $n_2'\geq n_2$, since in the latter case a subset of $n_1$ agents from the first group and a subset of $n_2$ agents from the second group still need to consider the allocation fair. Similarly, an existence result for $(n_1,n_2)$ yields a corresponding result for $(n_1',n_2')$ with $n_1'\leq n_1$ and $n_2'\leq n_2$.

We begin by considering the notions EFX and EF1. In fact, for binary valuations one can easily verify that EFX and EF1 are equivalent, so it suffices to consider only EF1.
We first present two results that establish the existence of an EF1 allocation for the cases $(n_1,n_2)=(5,1)$ and $(3,2)$. Before we present the proofs, we give a high-level overview of the arguments that we use. First, observe that since the valuations are binary, each good can be described by the set of agents who desire it. If two goods are desired by the same set of agents, we can allocate one to each group and then search for an EF1 allocation of the reduced instance with the remaining goods. Hence we may assume that every good is desired by a distinct subset of agents. This reduces the problem to a finite (but still large) number of possible instances. We then perform other preprocessing steps to reduce the number of cases even further. For example, if an agent desires an odd number of goods, the requirement that EF1 imposes on the agent remains the same when we perturb the valuation of the agent so that she no longer desires an arbitrary good. As a result, we may assume that every agent desires an even number of goods.

\begin{theorem}
\label{thm:binary-5-1}
For $(n_1,n_2)=(5,1)$ and binary valuations, an EF1 allocation always exists.
\end{theorem}

\begin{proof}
For the sake of convenience, we slightly modify the notation used in this proof. Consider an instance with two groups of agents $A = \{1, 2,3,4, 5\}$ and $B=\{b\}$. For a set of goods $G' \subseteq G$, let $d_i(G')$ be the subset of goods in $G'$ that agent $i$ desires. For binary valuations, each good $g \in G$ can be represented by the set of agents who desire it, i.e., we represent $g$ by the set $\{i\in A\cup B \mid g \in d_i(G) \}$. First, we perform a number of preprocessing steps as follows.

\begin{itemize}
\item[(P1)] We allocate each good $g \not\in d_b(G)$ to group $A$; this reduces the problem to finding an EF1 allocation of the remaining goods. Hence, we may assume that $d_b(G)=G$, and simply represent each good $g$ as a subset of $\{1,2,3,4,5\}$.

\item[(P2)] If there are an odd number of goods, then we allocate an arbitrary good to group $A$; this again reduces the problem to finding an EF1 allocation of the remaining goods. Hence, we may assume that there are an even number of goods.

\item[(P3)] If there are two goods $g$ and $\overline{g}$ such that $\overline{g} \subseteq g$, then we allocate $g$ to group $A$ and $\overline{g}$ to agent $b$. More generally, if there are two disjoint sets of goods $G_1$ and $G_2$ of equal size such that $|d_i(G_1)| \geq |d_i(G_2)|$ for every agent $i \in A$, then we allocate $G_1$ to group $A$ and $G_2$ to agent $b$. Therefore, in the following, we may assume that every good is represented by a distinct set of agents, and there is no subset-superset relations among the goods.

\item[(P4)] For any agent $i \in A$ that desires an odd number of goods, we assume a modified valuation function $\tilde{u}_i$ such that $\tilde{u}_i(g^*)=0$ for some $g^*$ that is desired by agent $i$ according to $u_i$, and $\tilde{u}_i(g)=u_i(g)$ for $g \neq g^*$. This again reduces the problem to finding an EF1 allocation for the instance with modified valuation functions. After we perform this step, it may be possible to perform the previous steps again. In such a case, we keep performing these steps until no longer possible.
\end{itemize}

After the preprocessing, there are an even number of goods and every agent desires an even number of goods. We claim that actually there are no goods left. Assume otherwise.
\begin{itemize}
\item If there is a good $g$ of size $0$ or $5$, then no other good $g'$ can appear since $g'$ would either be a superset or a subset of $g$, and step (P3) would have already allocated both $g$ and $g'$. Therefore, there is only one good, contradicting the assumption that the number of goods is even.

\item If there is a singleton good $g=\{i\}$ for some $i \in \{1,2,3,4,5\}$, then no other good $g'$ containing $i$ can appear, since otherwise $g'$ would be a superset of $g$. Therefore, agent $i$ only desires one good, contradicting the assumption that every agent desires an even number of goods.

\item If there is a good $g$ of size $4$, say $g=\{1,2,3,4\}$, then any other good $g'$ must contain $5$, since otherwise $g'$ would be a subset of $g$. Since there are an even number of goods, this means that agent $5$ desires an odd number of goods, again contradicting the assumption that all agents in $A$ desire an even number of goods.
\end{itemize}

Therefore, all goods correspond to subsets of size $2$ or $3$. Let $k_2$ and $k_3$ denote the number of goods of size $2$ and $3$, respectively. Since all agents desire an even number of goods, $k_3$ must be even, and since there are an even number of goods, $k_2$ must be even as well. It suffices to consider the case where $k_2 \geq k_3$. This is without loss of generality because if $k_2 < k_3$, we can replace each good $g$ by its complement $A \backslash g$. In the new instance, there are still an even number of goods, every agent still desires an even number of goods, and $k_2\geq k_3$. Moreover, if there is a preprocessing step that can be applied to the new instance, a corresponding step can be applied to the original instance.

 Recall that $|A|=5$ and every subset is represented by at most one good, so $k_2$ cannot exceed the number of sets of size $2$ from a universe of size $5$, which is $10$. Since $k_2$ is even, $k_2\geq k_3$, and $k_2+k_3>0$, we have $k_2\in\{2,4,6,8,10\}$. We consider each of these cases in turn.
\begin{itemize}
\item Case 1: $k_2=10$. This means that all goods of size $2$ are present. If we set $G_1=\{\{1,2\},\{3,4\}\}$ and $G_2=\{\{1,3\},\{2,4\}\}$, then $G_1$ and $G_2$ should have already been allocated in step (P3), a contradiction.

\item Case 2: $k_2= 8$. Consider the two goods of size $2$ that are absent. If they share one element, we may assume without loss of generality that they are $\{3,5\}$ and $\{4,5\}$; if they do not share any element, we may assume that they are $\{2,3\}$ and $\{4,5\}$. Either way, the goods $\{1,2\},\{1,3\},\{2,4\},\{3,4\}$ are present. If we set $G_1=\{\{1,2\},\{3,4\}\}$ and $G_2=\{\{1,3\},\{2,4\}\}$, then $G_1$ and $G_2$ should have already been allocated in step (P3), a contradiction.

\item Case 3: $k_2=6$. First, note that if $k_3\geq 2$, then since we have already allocated all goods that form subset-superset relations in step (P3), we must have $k_2\leq 5$. Hence, $k_3=0$. If at least two of the agents desire four goods each, there must be at least seven goods in total, which is not the case. Since every agent desires an even number of goods, the only possibility is that one of the agents desires four goods and the remaining four agents two goods each. Assume without loss of generality that the goods $\{1,2\},\{1,3\},\{1,4\},\{1,5\},\{2,3\},\{4,5\}$ are present. If we set $G_1=\{\{1,2\},\{1,3\},\{4,5\}\}$ and $G_2=\{\{1,4\},\{1,5\},\{2,3\}\}$, then $G_1$ and $G_2$ should have already been allocated in step (P3), a contradiction.

\item Case 4: $k_2=4$. Consider the four goods of size 2 that are present. Since the corresponding sets cannot be all disjoint, we may assume that two of them are $\{1,2\}$ and $\{1,3\}$. The subsets of size 3 that can be present are $\{1,4,5\},\{2,3,4\},\{2,3,5\},\{2,4,5\},\{3,4,5\}$.
\begin{itemize}
\item Case 4.1: Of the remaining two sets of size 2, at least one set contains 1. Assume without loss of generality that $\{1,4\}$ is present, so $\{1,4,5\}$ cannot be present. Since agent 1 desires an even number of goods, $\{1,5\}$ must be present. Using the constraint that each agent desires an even number of goods, one can check that $\{2,3,4\},\{2,3,5\},\{2,4,5\},\{3,4,5\}$ must all be present. However, sets $G_1=\{\{1,2\},\{3,4,5\}\}$ and $G_2=\{\{1,3\},\{2,4,5\}\}$ should have been allocated in step (P3), a contradiction.
\item Case 4.2: Of the remaining two sets of size 2, neither set contains 1. At least one of these sets must contain 2 or 3. Since agent 1 desires an even number of goods, $\{1,4,5\}$ cannot be present. Suppose first that $\{2,3\}$ is present, so $\{2,3,4\}$ and $\{2,3,5\}$ cannot be present. One can check that no matter whether $\{2,4,5\},\{3,4,5\}$ are both present or both absent, there is no way to make every agent desire an even number of goods. So $\{2,3\}$ is not present. Assume without loss of generality that $\{2,4\}$ is present, so $\{2,3,4\}$ and $\{2,4,5\}$ cannot be present. If $\{2,3,5\}$ and $\{3,4,5\}$ are both absent, the fourth set of size 2 must be $\{3,4\}$, and we may obtain a contradiction as in Case~2. Otherwise, $\{2,3,5\}$ and $\{3,4,5\}$ are both present, and the fourth set of size 2 must be $\{2,3\}$. But we have already deduced that $\{2,3\}$ must be absent, a contradiction.
\end{itemize}

\item Case 5: $k_2=2$. Consider the two goods of size 2 that are present. Since these goods have to be distinct and all agents desire an even number of goods, we cannot have $k_3=0$, so $k_3=2$.
\begin{itemize}
\item Case 5.1: These two goods share an agent. Assume without loss of generality that they are $\{1,2\}$ and $\{1,3\}$. Moreover, the two goods of size $3$ cannot be a superset of $\{1,2\}$ or $\{1,3\}$. The only possibility is that they are $\{2,4,5\}$ and $\{3,4,5\}$. If we set $G_1=\{\{1,2\},\{3,4,5\}\}$ and $G_2=\{\{1,3\},\{2,4,5\}\}$, then $G_1$ and $G_2$ should have already been allocated in step (P3), a contradiction.

\item Case 5.2: These two goods are disjoint. Assume without loss of generality that they are $\{1,2\}$ and $\{3,4\}$. Then the two goods of size $3$ must be either $\{1,3,5\}$ and $\{2,4,5\}$, or $\{1,4,5\}$ and $\{2,3,5\}$. In either case, we may set $G_1$ to contain the two goods of size $3$ and $G_2=\{\{1,2\},\{3,4\}\}$ and obtain a contradiction as before.
\end{itemize}
\end{itemize}
All cases have been exhausted, and the proof is complete.
\end{proof}

\begin{theorem}
\label{thm:binary-3-2}
For $(n_1,n_2)=(3,2)$ and binary valuations, an EF1 allocation always exists.
\end{theorem}

\begin{proof}
For the sake of convenience, we slightly modify the notation used in this proof. Consider an instance with two groups of agents $A=\{1,2,3\}$ and $B=\{x,y\}$. We represent each good by the set of agents who desire it. For convenience, we denote a good $g$ as $g=S^T$, where $S \subseteq A$ and $T \subseteq B$ are the subsets of agents who desire $g$ in $A$ and $B$, respectively.

Similarly to the proof of Theorem~\ref{thm:binary-5-1}, we can perform some preprocessing steps so that every good is desired by at least one agent in each group (if $g=S^T$ is such that $S=\emptyset$ or $T=\emptyset$, then we can simply allocate $g$ to group $B$ or group $A$, respectively), and every agent desires an even number of goods (if this is not the case, we modify the valuation function of every such agent by making the agent ``undesire'' an arbitrary good; this does not affect whether the agent considers an allocation to be EF1). We further assume that each agent in $B$ desires at least one good; otherwise the problem reduces to the $(n_1,n_2)=(3,1)$ case, which is already covered by Theorem~\ref{thm:binary-5-1}.

If two goods $g_1 = S_1^{T_1}$ and $g_2=S_2^{T_2}$ are such that $S_1 \supseteq S_2$ and $T_1 \subseteq T_2$, then we can allocate $g_1$ to $A$ and $g_2$ to $B$. Since the agents in $A$ that desire $g_2$ also desire $g_1$, and the agents in $B$ that desire $g_1$ also desire $g_2$, this reduces the problem to computing an EF1 allocation for the remaining goods. In the following we therefore assume that there do not exist two such goods; in particular, there is at most one good of type $S^T$ for each $S\subseteq\{1,2,3\}$ and $T\subseteq\{x,y\}$, and there are at most three goods of type $S^T$ for each $T \subseteq \{x,y\}$. More generally, if there are two disjoint sets of goods $G_1$ and $G_2$ such that every agent in $A$ desires at least as many goods in $G_1$ as in $G_2$, while every agent in $B$ desires at least as many goods in $G_2$ as in $G_1$, we can allocate $G_1$ to $A$ and $G_2$ to $B$. Consequently, we may assume that there do not exist two such sets of goods.

We now show that all goods have been allocated by the above preprocessing steps. By the above discussion, if the good $\{1,2,3\}^x$ is present, then no other good that agent $x$ desires can be present, which means that $x$ desires an odd number of goods, a contradiction. Hence, goods $\{1,2,3\}^x$ and $\{1,2,3\}^y$ are absent. Similarly, if the good $\{1\}^{xy}$ is present, then no other good that agent $1$ desires can be present, meaning that $1$ desires an odd number of goods, again contradicting the assumption that all agents desire an even number of goods. So, goods $\{1\}^{xy}$, $\{2\}^{xy}$, and $\{3\}^{xy}$ are absent. Therefore, the goods of type $S^{xy}$ that can be present are $\{1,2\}^{xy}$, $\{1,3\}^{xy}$, $\{2,3\}^{xy}$, and $\{1,2,3\}^{xy}$. We consider some cases.

\begin{itemize}
    \item Case 1: No good of type $S^{xy}$ is present. Since every agent desires an even number of goods and each agent in $B$ desires at least one good, there are exactly two goods of type $S^x$ and two goods of type $S^y$. Let these goods be  $S_1^x$, $S_2^x$, $T_1^y$ and $T_2^y$; by the above preprocessing steps, there are no subset-superset relations between $(S_1,S_2)$ and $(T_1,T_2)$.
    \begin{itemize}
        \item Case 1.1: $|S_1|=|S_2|=1$. Without loss of generality, assume that $S_1=\{1\}$ and $S_2=\{2\}$. Since each agent desires an even number of goods, the other two goods are either $\{1\}^y$ and $\{2\}^y$, or $\{1,3\}^y$ and $\{2,3\}^y$. In the first case we set $G_1=\{\{1\}^x,\{2\}^y\}$ and $G_2=\{\{1\}^y,\{2\}^x\}$, while in the second case we set $G_1=\{\{1\}^x,\{2,3\}^y\}$ and $G_2=\{\{1,3\}^y,\{2\}^x\}$. Then $G_1$ and $G_2$ should have been allocated during the preprocessing steps (since every agent desires the same number of goods in $G_1$ as in $G_2$), a contradiction.
        \item Case 1.2: $|S_1|=1$ and $|S_2|=2$. Without loss of generality, assume that $S_1=\{1\}$ and $S_2=\{2,3\}$. Since every agent desires an even number of goods, the other two goods can be  $(\{1\}^y,\{2,3\}^y)$, $(\{2\}^y,\{1,3\}^y)$, or $(\{3\}^y,\{1,2\}^y)$. We set $(G_1,G_2)$ to be $(\{\{1\}^x,\{2,3\}^y\},\{\{2,3\}^x,\{1\}^y\})$ in the first case, $(\{\{2,3\}^x,\{1,3\}^y\},\{\{1\}^x,\{2\}^y\})$ in the second case, and $(\{\{2,3\}^x,\{1,2\}^y\},\{\{1\}^x,\{3\}^y\})$ in the third case. Then, again, $G_1$ and $G_2$ should have been allocated during the preprocessing steps, a contradiction.
        \item Case 1.3: $|S_1|=|S_2|=2$. Without loss of generality, assume that $S_1=\{1,2\}$ and $S_2=\{1,3\}$. Since each agent desires an even number of goods, the other two goods are either $\{2\}^y$ and $\{3\}^y$, or $\{1,2\}^y$ and $\{1,3\}^y$. Then, similarly to the cases above, we set $G_1=\{\{1,2\}^x,\{3\}^y\}$ and $G_2=\{\{1,3\}^x,\{2\}^y\}$ in the first case, while in the second case we set $G_1=\{\{1,2\}^x,\{1,3\}^y\}$ and $G_2=\{\{1,3\}^x,\{1,2\}^y\}$. Then $G_1$ and $G_2$ should have been allocated during the preprocessing steps, a contradiction.
    \end{itemize}

    \item Case 2: Good $\{1,2,3\}^{xy}$ is present. Then, by the preprocessing steps, goods $\{1,2\}^{xy}$, $\{1,3\}^{xy}$, and $\{2,3\}^{xy}$ must be absent. Since every agent desires an even number of goods, there is an odd number (one or three) of goods of type $S^x$ as well as an odd number of goods of type $S^y$.
    \begin{itemize}
        \item Case 2.1: There is one good of type $S^x$. Since each agent desires an even number of goods, there is also exactly one good of type $T^y$ with $S\cap T=\emptyset$ and $S\cup T=\{1,2,3\}$. By setting $G_1=\{S^x,T^y\}$ and $G_2=\{\{1,2,3\}^{xy}\}$, $G_1$ and $G_2$ should have already been allocated during the preprocessing steps, a contradiction.
        \item Case 2.2: There are three goods of type $S^x$. These must be either $\{1\}^x,\{2\}^x,\{3\}^x$ or $\{1,2\}^x,\{1,3\}^x,\{2,3\}^x$. Suppose the former.
        Then, since every agent desires an even number of goods, goods $\{1,2\}^y, \{1,3\}^y, \{2,3\}^y$ must also be present. If we set $G_1=\{\{1\}^x,\{2,3\}^y\}$ and $G_2=\{\{1,3\}^y,\{2\}^x\}$, then $G_1$ and $G_2$ should have already been allocated during the preprocessing steps, a contradiction. The case where $\{1,2\}^x,\{1,3\}^x,\{2,3\}^x$ are present can be handled analogously.
    \end{itemize}

    \item Case 3: Good $\{1,2,3\}^{xy}$ is absent, and only one good of type $S^{xy}$ with $|S|=2$ is present. Without loss of generality, assume that this good is $\{1,2\}^{xy}$. Then, by the preprocessing steps, goods $\{1,2\}^{x}$ and $\{1,2\}^{y}$ must be absent. Since every agent desires an even number of goods, there is an odd number (one or three) of goods of type $S^x$ and an odd number of goods of type $S^y$. If there are three goods of type $S^x$, these must be $\{1\}^x,\{2\}^x,\{3\}^x$, but then, for every agent to desire an even number of goods, it has to be that the only additional good is $\{3\}^y$. In this case, however, if we set $G_1=\{\{1,2\}^{xy},\{3\}^x\}$ and $G_2=\{\{1\}^x,\{2\}^x, \{3\}^y\}$, then $G_1$ and $G_2$ should have already been allocated during the preprocessing steps, a contradiction.
        So, there is only one good of type $S^x$ and, similarly, only one good of type $T^y$. Moreover, we must have $1,2 \in S\cup T$ since, again, both agents $1$ and $2$ desire an even number of goods. If we set $G_1=\{S^x,T^y\}$ and $G_2=\{\{1,2\}^{xy}\}$, then $G_1$ and $G_2$ should have already been allocated during the preprocessing steps, a contradiction.

    \item Case 4: Good $\{1,2,3\}^{xy}$ is absent, and exactly two goods $S^{xy}$ with $|S|=2$ are present. Without loss of generality, assume that these are $\{1,2\}^{xy}$ and $\{1,3\}^{xy}$. By the preprocessing steps, goods $\{1,2\}^x$, $\{1,2\}^y$, $\{1,3\}^x$, and $\{1,3\}^y$ must be absent. Since every agent desires an even number of goods, there is an even number (zero or two) goods of type $S^x$ and an even number of goods of type $S^y$.
    \begin{itemize}
        \item Case 4.1: There is no good of type $S^x$ or no good of type $S^y$. Let us assume the former; the latter case can be handled analogously. Then, the goods of type $S^y$ must be $\{2\}^y$ and $\{3\}^y$, so that every agent in $A$ desires two goods. If we set $G_1=\{\{1,2\}^{xy},\{3\}^y\}$ and $G_2=\{\{1,3\}^{xy},\{2\}^y\}$, then $G_1$ and $G_2$ should have already been allocated during the preprocessing steps, a contradiction.
        \item Case 4.2: There are two goods of type $S^x$ and two goods of type $S^y$. If $\{2,3\}^x$ is present, then goods $\{2\}^x$ and $\{3\}^x$ must be absent by the preprocessing steps. However, this means that good $\{1\}^x$ must be present (so that there are two goods of type $S^x$), but there is no way we can pick two goods of type $S^y$ so that all agents desire an even number of goods. Hence, good $\{2,3\}^x$ must be absent and, similarly, $\{2,3\}^y$ must also be absent. Therefore, we must have goods $\{1\}^x$, $\{1\}^y$, $S^x$, and $T^y$, where $S,T \in \{\{2\},\{3\}\}$ and $S\neq T$. If we set $G_1=\{\{1\}^x,\{1\}^y,S^x,T^y\}$ and $G_2=\{\{1,2\}^{xy},\{1,3\}^{xy}\}$, then $G_1$ and $G_2$ should have already been allocated during the preprocessing steps, a contradiction.
    \end{itemize}

    \item Case 5: Good $\{1,2,3\}^{xy}$ is absent, and all three goods $S^{xy}$ with $|S|=2$ are present. By the preprocessing steps, goods $\{1,2\}^{x}$, $\{1,2\}^{y}$, $\{1,3\}^x$, $\{1,3\}^y$, $\{2,3\}^x$, and $\{2,3\}^y$ must be absent. Since every agent desires an even number of goods, there is an odd number (one or three) goods of type $S^x$ and an odd number of goods of type $S^y$. The only possibility that makes each agent of group $A$ desire an even number of goods is that there are three goods of each type: $\{1\}^x,\{2\}^x,\{3\}^x$ and $\{1\}^y,\{2\}^y,\{3\}^y$. If we set $G_1=\{\{1\}^x,\{2\}^x,\{3\}^x,\{1\}^y,\{2\}^y,\{3\}^y\}$ and $G_2=\{\{1,2\}^{xy},\{1,3\}^{xy},\{2,3\}^{xy}\}$, then $G_1$ and $G_2$ should have already been allocated during the preprocessing steps, a contradiction.
\end{itemize}
All cases have been exhausted, and the proof is complete.
\end{proof}

The following results complete the characterization for EF1 (and EFX) by proving that existence is not guaranteed for larger sets.

\begin{proposition}
\label{prop:binary-6-1}
For $(n_1,n_2)=(6,1)$ and binary valuations, an EF1 allocation does not always exist.
\end{proposition}

\begin{proof}
Suppose that there are four goods. For each pair of goods, there is an agent in the first group who desires both of the goods in the pair but does not desire the remaining two goods. On the other hand, the singleton agent in the second group desires all goods. To guarantee EF1, this agent must receive at least two goods, leaving at most two goods for the first group. However, this leaves some agent in the first group with utility 0, and such an agent does not find the allocation to be EF1.
\end{proof}

\begin{proposition}
\label{prop:binary-4-2}
For $(n_1,n_2)=(4,2)$ and binary valuations, an EF1 allocation does not always exist.
\end{proposition}

\begin{proof}
Suppose that there are four goods. The utilities in the first group are given by $\textbf{u}_{11}=(1,0,1,0)$, $\textbf{u}_{12}=(1,0,0,1)$, $\textbf{u}_{13}=(0,1,1,0)$, and $\textbf{u}_{14}=(0,1,0,1)$, and the utilities in the second group by $\textbf{u}_{21}=(1,1,0,0)$ and $\textbf{u}_{22}=(0,0,1,1)$. In an EF1 allocation, every agent needs at least one desired good. In particular, the second group needs one of the first two goods and one of the last two goods. However, any choice of these goods leaves some agent in the first group with utility 0, meaning that no allocation in EF1.
\end{proof}

In the next proposition, we give a general example for EF$c$ since this will be useful later on.

\begin{proposition}
\label{prop:binary-equal}
Let $c$ be a positive integer. For $(n_1,n_2)=\left(\binom{2c+1}{c+1},\binom{2c+1}{c+1}\right)$ and binary valuations, an EF$c$ allocation does not always exist.
\end{proposition}

\begin{proof}
Suppose that there are $2c+1$ goods. For each subset of $c+1$ goods, there is exactly one agent in each group who desires the goods in this subset and nothing else. Every agent must get at least one desired good in an EF$c$ allocation. However, in any allocation one of the groups receives at most $c$ goods. In that group, at least one of the agents does not get any desired good. Hence no allocation can be EF$c$.
\end{proof}

Taking $c=1$ in Proposition~\ref{prop:binary-equal} yields the following:

\begin{corollary}
\label{cor:binary-3-3}
For $(n_1,n_2)=(3,3)$ and binary valuations, an EF1 allocation does not always exist.
\end{corollary}

Before addressing EFX$_0$, we show that for two groups of arbitrary sizes, determining whether an EF1 (and EFX) allocation exists is computationally hard. Our reduction is similar to the one used by \citet{SegalhaleviSu18} to show the hardness of deciding the existence of an allocation that gives every agent a positive utility.

\begin{proposition}
\label{prop:binary-hard}
For two groups of agents with binary valuations, it is NP-complete to decide whether there exists an EF1 allocation.
\end{proposition}

\begin{proof}
The problem is in NP, since we can clearly verify in polynomial time whether an allocation is EF1 for every agent. To show hardness, we use a reduction from the {\sc Monotone 3-SAT} problem, which asks whether a Boolean formula consisting of clauses with either three positive or three negative literals can be satisfied. In fact, we will use the slightly stronger version where the three literals in each clause are distinct; this version is still NP-complete \citep{Li97}.

Given an instance of {\sc Monotone 3-SAT} consisting of a formula $\phi$, we construct an instance of our problem with two groups of agents $A_1$ and $A_2$ as follows:
\begin{itemize}
\item For each variable, there is a corresponding good.
\item For each clause with three positive literals, we create an agent in $A_1$ who desires exactly the three goods corresponding to the variables contained in the clause.
\item For each clause with three negative literals, we create an agent in $A_2$ who desires exactly the three goods corresponding to the variables contained in the clause.
\end{itemize}

Note that every agent needs at least one desired good in order to be EF1. Any assignment that satisfies $\phi$ defines an allocation where the goods corresponding to true variables are allocated to group $A_1$, while those corresponding to false variables are allocated to group $A_2$. Since all clauses are satisfied, every agent receives utility at least 1 and is therefore EF1. Similarly, any EF1 allocation gives rise to a satisfying assignment for $\phi$, completing the reduction.
\end{proof}

We now turn to the stronger notion of EFX$_0$, and show a negative result.

\begin{proposition}
\label{prop:binary-2-1}
For $(n_1,n_2)=(2,1)$ and binary valuations, an EFX$_0$ allocation does not always exist.
\end{proposition}

\begin{proof}
Suppose that there are six goods. The utilities are given by $\textbf{u}_{11}=(1,1,1,0,0,0)$, $\textbf{u}_{12}=(0,0,0,1,1,1)$, and $\textbf{u}_{21}=(1,1,1,1,1,1)$. To guarantee EFX$_0$, the singleton agent in the second group must receive at least three goods, leaving at most three goods for the first group. This means that at least one of the agents in the first group, say $a_{11}$, receives at most one desired good. If $a_{11}$ does not receive any desired good, the allocation is clearly not EFX$_0$ for her. Else, she receives exactly one desired good. In this case, she has utility 2 for the second group's bundle, and at least one of the goods in that bundle yields utility 0 to her. Therefore the allocation cannot be EFX$_0$.
\end{proof}

In contrast, \citet{PlautRo18} showed that an EFX$_0$ allocation always exists if $(n_1,n_2)=(1,1)$, even for arbitrary monotonic valuations. Combined with Proposition~\ref{prop:binary-2-1}, this yields a complete characterization of EFX$_0$ for every class of valuations between binary and monotonic.

\section{Fixed Groups with General Valuations}
\label{sec:fixed-general}

In this section, we again assume that the partition of the agents into two groups is predetermined, but allow them to have more general valuations.

We first show that the existence of EF1 allocations is guaranteed for $(n_1,n_2)=(2,1)$. The proof relies on the following lemma which may be of independent interest. A partition of the goods in $G$ into two bundles is said to be \emph{exact up to one good (Exact1)} for an agent if the agent views each bundle to be EF1. As with allocations, we call a partition of the goods \emph{balanced} if the sizes of the two bundles differ by at most one.
The proof of the lemma involves constructing a bipartite graph with the goods as vertices, considering a partition of the vertices into two disjoint independent sets, and showing that this partition is necessarily Exact1.

\begin{lemma}
\label{lem:Exact1}
For two agents with responsive valuations, there always exists a balanced partition of $G$ into two bundles that is Exact1 for both agents.
\end{lemma}

\begin{proof}
Suppose first that the number of goods $m$ is even, say $m=2t$. Assume without loss of generality that the valuation $u_1$ of the first agent is such that $u_1(g_1)\geq u_1(g_2)\geq\dots\geq u_1(g_{2t})$. Construct an undirected graph $H$ with $2t$ vertices corresponding to the goods, and add $t$ red edges $(g_1,g_2),(g_3,g_4),\dots,(g_{2t-1},g_{2t})$. Similarly, add $t$ blue edges according to the valuation $u_2$ of the second agent. Since no two edges of the same color are adjacent, the graph cannot contain an odd cycle, which means that $H$ is bipartite. Therefore, its vertices can be partitioned into disjoint independent sets $V_1$ and $V_2$. If $|V_1|\geq t+1$, there is an edge among the vertices in $V_1$, a contradiction. An analogous statement holds for $V_2$. It follows that $|V_1|=|V_2|=t$.

The partition $(V_1,V_2)$ is balanced; it remains to show that it is Exact1 for both agents. By symmetry, it suffices to prove this for the first agent. By construction, each of $V_1$ and $V_2$ contains exactly one good from each of the pairs $(g_1,g_2),(g_3,g_4),\dots,(g_{2t-1},g_{2t})$. For $i=1,2,\dots,t-1$, the $i$th best good in $V_1$ according to the first agent's valuation is no worse than the $(i+1)$st best good in $V_2$. Responsiveness then implies that the agent values $V_1$ at least as much as $V_2$ when the best good in $V_2$ is removed. This means that she regards $V_1$ to be EF1. Similarly, she regards $V_2$ to be EF1; hence the partition is Exact1 for her.

Suppose now that $m=2t-1$ is odd. We add a dummy good $g_{2t}$ such that $u_i(G'\cup\{g_{2t}\}) = u_i(G')$ for $i=1,2$ and any $G'\subseteq G$. We then repeat the same procedure as in the case where $m$ is even (placing $g_{2t}$ at the end of each agent's ranking of single goods), and remove $g_{2t}$ from the resulting partition. The final partition is balanced, and a similar proof as before shows that it is Exact1 for both agents.
\end{proof}

Note that the lemma no longer holds if we move to three agents (while still keeping the partition into two bundles), even for binary valuations. Indeed, if there are three goods and each of the three agents desires a distinct subset of two goods, then no partition is Exact1 for all three agents.
Recently, \citet{GoldbergHoIg20} showed that an Exact1 partition always exists even for arbitrary monotonic valuations, although the partition that their algorithm returns may not be balanced.

Lemma~\ref{lem:Exact1} yields the following EF1 existence result.

\begin{theorem}
\label{thm:responsive-2-1}
For $(n_1,n_2)=(2,1)$ and responsive valuations, a balanced EF1 allocation always exists.
\end{theorem}

\begin{proof}
Choose two arbitrary agents and consider a balanced partition of $G$ into two bundles that is Exact1 for both agents; such a partition exists by Lemma~\ref{lem:Exact1}. Let the remaining agent choose for her group the bundle that she prefers, and allocate the other bundle to the other group. It is clear that the resulting allocation is balanced and EF1.
\end{proof}

In light of Theorem~\ref{thm:responsive-2-1} and our characterization for binary valuations in Section~\ref{sec:fixed-binary}, it is natural to ask whether EF1 can also be guaranteed for larger groups with additive valuations and beyond. While we were unable to settle this question, we show that the existence of EF1 allocations would be guaranteed for almost all of the remaining cases provided that a graph-theoretic conjecture of \citet{JafariAl17} is true. To describe the conjecture and its implications in our setting, we need to introduce a class of graphs called \emph{generalized Kneser graphs}.\footnote{Kneser graphs have previously been used in the context of fair division and resource allocation by \citet{PlautRo18} and \citet{ManurangsiSu19}.} 

\begin{definition}
Let $b\geq r\geq s$ be positive integers and consider an underlying set of elements $\mathcal{U}$ such that $|\mathcal{U}|=b$. The \emph{generalized Kneser graph $K(b,r,s)$} is an undirected graph with all $r$-element subsets of $\mathcal{U}$ as its vertices.
Two vertices are connected by an edge if and only if the corresponding subsets intersect in at most $s-1$ elements.
\end{definition}

Recall that the \emph{chromatic number} of a graph $H$, denoted by $\chi(H)$, is the minimum number of colors needed to color the vertices of $H$ so that any two adjacent vertices have different colors. For example, $K(4,2,2)$ is a clique of size 6, so its chromatic number is 6.

We are now ready to establish the connection between the generalized Kneser graph and our fair division problem.
In the following proof, we let the vertices of the graph represent the balanced allocations, and have each agent color all allocations that she does not regard as EF1. 
We then show that some vertex must be left uncolored, which implies that the corresponding allocation is EF1 for all agents.

\begin{theorem}
\label{thm:monotonic-EF1}
Let $t\geq 2$ be an integer. For two groups with at most $\chi(K(2t,t,2))-1$ agents in total and arbitrary monotonic valuations, and $2t-1$ or $2t$ goods, a balanced EF1 allocation always exists.
\end{theorem}

\begin{proof}
Suppose first that the number of goods is $m=2t$.  Consider the graph $K(2t,t,2)$ with the vertices corresponding to all balanced allocations, where we identify each vertex by the set of goods allocated to the first group.

Give each agent a distinct color, and let her color all allocations that she \emph{does not} regard as EF1. We claim that no agent can color two adjacent vertices. Consider an agent in the first group with valuation $u$, and suppose for contradiction that she colors two adjacent vertices corresponding to the sets $G_1$ and $G_2$. Since the two vertices are adjacent, $|G_1\cap G_2|\leq 1$.
If $G_1\cap G_2=\emptyset$, since $|G_1|=|G_2|=t$ it holds that $G_2=G\backslash G_1$. So the agent should consider one of the two allocations as EF, which is a contradiction to her coloring both vertices.
Therefore $|G_1\cap G_2|=1$. Let $g$ be the common good of $G_1$ and $G_2$. Since the agent does not view $G_1$ to be EF1, we have $u(G_1)<u(G_2\backslash\{g\})=u((G\backslash G_1)\backslash \{g'\})$, where $g'$ is the unique good that does not belong to $G_1\cup G_2$. Similarly, since the agent does not view $G_2$ to be EF1, we have $u(G_2)<u(G_1\backslash\{g\})$. Monotonicity then implies that $$u(G_1)<u(G_2\backslash\{g\})\leq u(G_2)<u(G_1\backslash\{g\})\leq u(G_1),$$ a contradiction. The claim can be proven similarly for agents in the second group by observing that for any two balanced allocations, the bundles allocated to the first group intersect in at most one good if and only if the same condition holds for the bundles allocated to the second group.

Since there are at most $\chi(K(2t,t,2))-1$ agents, the number of colors is at most $\chi(K(2t,t,2))-1$. Hence there is a vertex that does not receive any color. By definition, this vertex corresponds to a balanced allocation that is EF1 for all agents. This completes the proof for the case where $m=2t$.

Suppose now that $m=2t-1$. We add a dummy good $g_{2t}$ such that $u_{ij}(G'\cup\{g_{2t}\}) = u_{ij}(G')$ for all agents $a_{ij}$ and all $G'\subseteq G$. We then repeat the same procedure as in the case where $m=2t$ and remove $g_{2t}$ from the resulting allocation. The final allocation is balanced, and a similar proof as before shows that it is EF1.
\end{proof}

The bound $\chi(K(2t,t,2))-1$ in Theorem~\ref{thm:monotonic-EF1} is also tight: for two groups with $y:=\chi(K(2t,t,2))$ agents in total and $2t$ goods, there does not always exist a balanced EF1 allocation.
This can be seen as follows.
Consider a coloring of $K(2t,t,2)$ using $y$ colors such that no two adjacent vertices have the same color---such a coloring exists due to the definition of $y$.
Distribute the $y$ agents arbitrarily between the two groups, and assign to each agent a distinct color.
It suffices to show that for each color, there exists a monotonic valuation for the corresponding agent such that the allocations of this color are precisely those that the agent does not regard as EF1.
Focus on a particular agent with valuation $u$ and her corresponding color, and identify each vertex of the graph by the bundle allocated to the agent's group.
As noted in the proof of Theorem~\ref{thm:monotonic-EF1}, for any two balanced allocations, the bundles allocated to the first group intersect in at most one good exactly when the same holds for those allocated to the second group, so it suffices to consider the case where the agent belongs to the first group.
For each colored bundle $G'$, we set $u(G'')=0$ for every $G''\subseteq G'$.
For all remaining bundles $G''$, we set $u(G'')=1$.
Clearly $u$ is monotonic, and all uncolored bundles (of size $t$) have value $1$, which means that they are EF1 for the agent.
It remains to show that all colored bundles fail EF1.
Let $G_1$ be a colored bundle, so $u(G_1)=0$.
Since no two adjacent vertices are both colored, every bundle $G_2$ of size $t$ which has an overlap of at most one good with $G_1$ is not colored, so $u(G_2)=1$ for such $G_2$.
By definition of $u$, this means that for every bundle $G_3$ of size $t-1$ disjoint from $G_1$, we also have $u(G_3)=1$.
It follows that $G_1$ is not EF1 for the agent, completing the proof of tightness.

We remark that our proof of Theorem~\ref{thm:monotonic-EF1} does not imply a polynomial-time algorithm for computing a balanced EF1 allocation.
Indeed, as problems involving Kneser graphs tend to be computationally hard (see, e.g., \citep{PlautRo18}), we do not expect that our approach will lead to an efficient algorithm.
Finding an alternative approach that lends itself to efficient algorithms is an interesting open direction.

We now state the existence of balanced EF1 allocations for any number of goods.

\begin{corollary}
\label{cor:monotonic-EF1}
Let $z:=\min_{r\geq 2}\chi(K(2r,r,2))$. For two groups with at most $z-1$ agents in total and arbitrary monotonic valuations, a balanced EF1 allocation always exists.
\end{corollary}

\begin{proof}
Suppose first that the number of goods $m$ is even, say $m=2t$. 
If $t=1$, any allocation that gives one good to each group is balanced and EF1, so we may assume that $t\geq 2$.
Since the number of agents is at most $z-1\leq \chi(K(2t,t,2))-1$, Theorem~\ref{thm:monotonic-EF1} implies that a balanced EF1 allocation exists.
A similar proof holds for the case where $m$ is odd.
\end{proof}

\citet{JafariAl17} proved that $K(2r,r,2)\leq 6$ for all $r\geq 2$, and conjectured that this bound is in fact always tight.

\begin{conjecture}[\citep{JafariAl17}]
\label{conj:chromatic}
For any $r\geq 2$, we have $\chi(K(2r,r,2))=6$.
\end{conjecture}

If Conjecture~\ref{conj:chromatic} is true, then together with Corollary~\ref{cor:monotonic-EF1}, it would imply that a balanced EF1 allocation is guaranteed to exist for two groups with at most 5 agents in total and arbitrary monotonic valuations.\footnote{\citet{JafariAl17} also claimed that $\chi(K(2r,r,2))\geq 4$ for all $r\geq 2$. In combination with our Corollary~\ref{cor:monotonic-EF1}, this would imply that our Theorem~\ref{thm:responsive-2-1} can be generalized to arbitrary monotonic valuations. However, the proof of their Theorem~5.1 contains an error when they claim that the intersection of the two $k$-subsets of color $j$ has at most $i-1$ elements. It is only true that the intersection has at most $i-1$ elements \emph{in each of the two hemispheres}, and therefore at most $2i-2$ elements in total. It is not clear whether the proof can be recovered in light of this error.} The bound of 5 cannot be improved to 6 due to Corollary~\ref{cor:binary-3-3}. Moreover, this result would answer the EF1 existence question in the affirmative for all of the remaining group sizes except for the case $(5,1)$. We remark that for this case, a \emph{balanced} EF1 allocation might not exist, even when valuations are binary.

\begin{proposition}
\label{prop:binary-balanced-5-1}
For $(n_1,n_2)=(5,1)$ and binary valuations, a balanced EF1 allocation does not always exist.
\end{proposition}

\begin{proof}
Suppose that there are four goods. The utilities in the first group are given by $\textbf{u}_{11}=(1,1,0,0)$, $\textbf{u}_{12}=(1,0,1,0)$,
$\textbf{u}_{13}=(1,0,0,1)$, $\textbf{u}_{14}=(0,1,1,0)$,
$\textbf{u}_{15}=(0,1,0,1)$, and in the second group by $\textbf{u}_{21}=(1,1,0,0)$. The only balanced allocation that is EF1 for all agents in the first group is the allocation that gives the first two goods to the first group. However, this allocation is not EF1 for the singleton agent, so no balanced allocation is EF1.\footnote{This instance does, however, admit an EF1 allocation. For example, the allocation that gives only the first good to the singleton agent is EF1.}
\end{proof}

We note that our techniques in Theorem~\ref{thm:monotonic-EF1} and Corollary~\ref{cor:monotonic-EF1} can be extended to weaker relaxations of envy-freeness. 
For any positive integer $c$, it is known that an EF$c$ allocation is guaranteed to exist for two groups with at most $c+1$ agents in total and additive valuations \citep{SegalhaleviSu18}. On the other hand, letting $z_c:=\min_{r\geq c+1}\chi(K(2r,r,c+1))$, a similar proof as in Theorem~\ref{thm:monotonic-EF1} and Corollary~\ref{cor:monotonic-EF1} shows that an EF$c$ allocation can always be found when the two groups contain at most $z_c-1$ agents in total with arbitrary monotonic valuations. \citet{JafariAl17} proved that $\chi(K(2r,r,c+1))\leq\binom{2c+2}{c+1}$ for all $r\geq c+1$. If this inequality becomes equality for all $r$ (in which case we would have $z_c=\binom{2c+2}{c+1}$), it would imply an exponential improvement in the relation between the number of agents and the number of goods in the EF$c$ approximation. The bound $z_c-1$ would also be tight due to the instance in Proposition~\ref{prop:binary-equal}.

We conclude this section by showing that existence can no longer be guaranteed if we strengthen the fairness requirement from EF1 to EFX. Recall that for binary valuations, an EFX allocation always exists when one group contains at most five agents and the other group is a singleton. We show that this is not the case for additive valuations, even when the first group contains only two agents.

\begin{proposition}
\label{prop:additive-EFX}
For $(n_1,n_2)=(2,1)$ and additive valuations, an EFX allocation does not always exist.
\end{proposition}

\begin{proof}
Suppose that there are four goods. The utilities are given by $\textbf{u}_{11}=(3,1,1,1)$, $\textbf{u}_{12}=(1,3,1,1)$, and $\textbf{u}_{21}=(3,3,1,1)$. In an EFX allocation, the singleton agent in the second group needs either both of the first two goods, or one of the first two and at least one of the last two goods. The former option leaves both agents in the first group unsatisfied. For the latter option, assume without loss of generality that the singleton agent receives $g_1$ and $g_3$. Then the resulting allocation is not EFX for $a_{11}$. Hence there is no EFX allocation in this instance.
\end{proof}

Since an EFX allocation always exists when $(n_1,n_2)=(1,1)$ for arbitrary monotonic valuations \citep{PlautRo18}, we have a complete characterization of EFX for every class of valuations between additive and monotonic.

\section{Variable Groups}
\label{sec:variable}

Thus far, we have worked under the assumption that the partition of agents into groups is determined in advance. This assumption is appropriate when we consider, for example, membership in a family or citizenship of a country. In other settings, however, the choice of the group to which the agents belong can be made by a central authority or by the agents themselves. This applies to membership in a library, gym, or other facilities.

With this motivation in mind, we depart from the framework of fixed groups in this section, and instead assume that the partition of the agents into groups can be chosen along with the allocation of the goods. Under this assumption, finding an EF1, EFX, or even envy-free allocation is trivial: simply put all agents in one group and allocate all goods to that group. However, this may lead to undesirable situations where a gym is overcrowded or a library does not have enough space to hold all of its books. As we will show, it is nevertheless possible to obtain a fair outcome that is moreover balanced with respect to both the agents and the goods, for any number of agents with general valuations.

\subsection{Two Groups}
\label{sec:variable-two}

We start with two groups and show that EF1 can be guaranteed for any desired sizes of these groups.
Our algorithm generalizes the discrete ``cut-and-choose'' algorithm for allocating indivisible goods between two individual agents \citep{BiloCaFl19,OhPrSu19}.

\begin{theorem}
\label{thm:EF1-balanced-agent}
Let $n$ be any positive integer. Suppose that there are $n$ agents with arbitrary monotonic valuations, and let $n_1$ and $n_2$ be non-negative integers with $n_1+n_2=n$.  There always exists a partition of the agents into two groups such that group $i\in\{1,2\}$ contains $n_i$ agents, along with an EF1 allocation of the goods to the two groups.
\end{theorem}

\begin{proof}
Arrange the goods in a line. Starting with an empty bundle, we add one good at a time from the left until at least $n_1$ agents find the bundle to be EF1. If this condition is met before we add any good, we give $n_1$ of these agents an empty bundle and the remaining $n_2$ agents the entire set $G$. Otherwise, denote by $g$ the last good added to the bundle. We assign to the first group all agents who view the bundle as EF1 before the addition of $g$, along with an arbitrary subset of those who find it EF1 after $g$ is added so that the first group has size $n_1$. We allocate this bundle to the first group, and the remaining goods to the second group, which consists of the remaining agents.

Since the entire set $G$ is EF1 for all agents, the process terminates. By construction, the agents in the first group regard the allocation as EF1, so we only need to show that the same holds for the second group. This is trivial if the second group receives the entire set $G$. Otherwise, let $G_L$ and $G_R$ be the bundles to the left and right of $g$, respectively (not containing $g$). Every agent in the second group does not find $G_L$ to be EF1, which means that she has more value for $G_R$ than $G_L$. This implies that the agent finds $G_R$ to be EF1, as desired.
\end{proof}

For the case where a balanced allocation of the goods is required, we prove that this can also be achieved and, in fact, it can always be combined with a balanced partition of the agents.
This means that in our gym and library applications, it is possible to reach a balance in terms of the users as well as the resources.
Our algorithm arranges the goods in a circle and moves a knife around the center of the circle, so all partitions of the goods that we consider are balanced.
The crux of our proof lies in showing that at some point during this process, the numbers of agents who prefer either part of the partition will be roughly equal.

\begin{theorem}
\label{thm:EF1-balanced-agent-good}
Let $n$ be any positive integer, and suppose that there are $n$ agents with arbitrary monotonic valuations. There always exists a balanced partition of the agents into two groups along with a balanced EF1 allocation of the goods to the two groups.
\end{theorem}

\begin{proof}
Suppose first that both the number of agents $n$ and the number of goods $m$ are even, say $n=2s$ and $m=2t$. Assume for contradiction that there is no balanced partition along with a balanced allocation.

Arrange the goods around a circle with equal spacing between adjacent goods. Imagine a knife that cuts through the center of the circle, dividing the goods into two bundles $G_1$ and $G_2$, each of size $t$. By our assumption, any balanced assignment of the agents to $G_1$ and $G_2$ does not result in an EF1 allocation. On the other hand, there does exist an assignment such that the resulting allocation is EF1 (e.g., an assignment that gives every agent her favorite bundle). Consider such an assignment that moreover minimizes the difference between the numbers of agents in the two groups. Assume without loss of generality that more than half of the agents are assigned to $G_1$. If one of these agents finds $G_2$ to be EF1, we can reassign her to $G_2$ and reduce the discrepancy between the two groups. Hence all of these agents do not find $G_2$ to be EF1.

Next, we rotate the knife clockwise by one position, thereby moving a good $g$ from $G_1$ to $G_2$ and another good $\overline{g}$ from $G_2$ to $G_1$. Call the resulting bundles $H_1=(G_1\cup\{\overline{g}\})\backslash \{g\}$ and $H_2=(G_2\cup\{g\})\backslash \{\overline{g}\}$, respectively. We claim that the agents who do not find $G_2$ to be EF1 regard $H_1$ as EF1. Denoting the valuation of an arbitrary such agent by $u$, we have
\begin{align*}
u(H_1)&\geq u(G_1\backslash\{g\}) \\
&>u(G_2) \geq u(G_2\backslash\{\overline{g}\})=u(H_2\backslash\{g\}),
\end{align*}
where the first and third inequalities follow from monotonicity. So the agent indeed finds $H_1$ to be EF1. Since more than half of the agents do not find $G_2$ to be EF1, more than half of the agents regard $H_1$ as EF1. By our assumption, any balanced assignment of the agents to $H_1$ and $H_2$ does not result in an EF1 allocation. It follows that in any assignment such that the resulting allocation is EF1, more than half of the agents are assigned to $H_1$. Additionally, more than half of the agents do not find $H_2$ to be EF1.

If we rotate the knife clockwise repeatedly, the same argument tells us that more than half of the agents do not find the second bundle (i.e., $G_2$, $H_2$, and so on) to be EF1. After $t$ rotation steps, the knife has rotated halfway around the circle, and the second bundle coincides with the original first bundle $G_1$. However, we know from earlier that more than half of the agents find $G_1$ to be EF1. This yields the desired contradiction.

Suppose now that $n$ is odd or $m$ is odd (or both). If $n$ is odd, we add a dummy agent with an arbitrary monotonic valuation. If $m$ is odd, we add a dummy good that always yields zero marginal utility for every agent. We then repeat the same procedure as in the case where $m$ and $n$ are even, and remove the dummy agent and/or the dummy good. The resulting partition and allocation are both balanced, and a similar proof as before shows that the allocation is EF1.
\end{proof}

Theorem~\ref{thm:EF1-balanced-agent-good} yields the following result on individual fair division, which is new to the best of our knowledge.

\begin{corollary}
\label{cor:EF1-individual-balanced}
For two individual agents with arbitrary monotonic valuations, there always exists a balanced EF1 allocation.
\end{corollary}

When valuations are responsive, a balanced EF1 allocation (for arbitrarily many agents) can also be obtained by the round-robin algorithm, which lets the agents take turns choosing their favorite good from the remaining goods until all goods are taken (see, e.g., \citep{CaragiannisKuMo16}). However, the round-robin algorithm does not work for arbitrary monotonic valuations.

Turning our attention to EFX, we show that if we require the partition of the agents to be balanced, an EFX allocation might not exist; this complements Theorems~\ref{thm:EF1-balanced-agent} and \ref{thm:EF1-balanced-agent-good} above. In addition, a balanced EFX allocation does not necessarily exist for two individual agents even when the agents have identical additive valuations, which complements Corollary~\ref{cor:EF1-individual-balanced}.

\begin{proposition}
\label{prop:EFX-balanced-agent}
There does not always exist a balanced partition of the agents into two groups along with an EFX allocation of the goods to the two groups, even when valuations are additive.
\end{proposition}

\begin{proof}
Suppose that there are six agents and three goods. The utilities of the agents are given by $\textbf{u}_1=\textbf{u}_2=(3,1,1)$, $\textbf{u}_3=\textbf{u}_4=(1,3,1)$, and $\textbf{u}_5=\textbf{u}_6=(1,1,3)$. If the allocation places all three goods in one bundle, all agents must be assigned to that bundle for the allocation to be EFX. Else, the allocation places two goods in one bundle and one good in the other bundle. In this case, at most two agents can be assigned to the second bundle in an EFX allocation, so the partition of agents would not be balanced. Hence there is no balanced partition of agents if the allocation must be EFX.
\end{proof}

\begin{proposition}
\label{prop:EFX-individual-balanced}
Let $m$ be a positive integer. There exists an instance with two individual agents who have identical additive valuations and $m$ goods, such that in any EFX allocation, one of the agents receives exactly one good.
\end{proposition}

\begin{proof}
Consider two agents with the same additive valuation $u(g_1)=m$ and $u(g_i)=1$ for $i=2,3,\dots,m$. Assume without loss of generality that the first agent receives $g_1$. If the first agent also receives another good $g_i$, then the second agent's value for the first agent's bundle when $g_i$ is removed is at least $m$. On the other hand, the second agent's value for her own bundle is at most $m-2$. Hence the first agent cannot receive another good besides $g_1$.
\end{proof}

\subsection{Any Number of Groups}
\label{sec:variable-many}

We now consider any number of groups and show how Theorem~\ref{thm:EF1-balanced-agent} can be partially extended to this setting. While we do not know whether EF1 or other relaxations of envy-freeness can be achieved in this general setting, we show that we can obtain a positive result for a weaker fairness notion called \emph{proportionality}. An allocation of the goods in $G$ to $k$ groups is said to be \emph{proportional} if every agent receives value at least $1/k$ of her value for the whole set $G$. As with envy-freeness, a proportional allocation does not always exist (e.g., when there are two individual agents and one valuable good), so it is necessary to consider a relaxation. Let $u_{j,\text{max}}:=\max_{t=1}^m u_j(g_t)$ denote the maximum value of agent $a_j$ for a single good.

\begin{theorem}
\label{thm:proportional-many-groups}
Let $n$ and $k$ be any positive integers. Suppose that there are $n$ agents with additive valuations, and let $n_1,\dots,n_k$ be non-negative integers with $\sum_{i=1}^k n_i=n$. There always exists a partition of the agents into $k$ groups such that group $i$ contains $n_i$ agents, along with an allocation of the goods to the $k$ groups such that each agent $a_j$ receives utility at least $\frac{1}{k}\cdot u_j(G) - \frac{k-1}{k}\cdot u_{j,\text{max}}$.
\end{theorem}

\begin{proof}
Let $w_j:=\frac{1}{k}\cdot u_j(G) - \frac{k-1}{k}\cdot u_{j,\text{max}}$. Arrange the goods in a line and process them from left to right. Suppose that we have already formed $i-1$ groups. Starting with an empty bundle, we add one good at a time from the left until at least $n_i$ agents $a_j$ receive utility at least $w_j$. If this condition is met before we add any good, assign $n_i$ of these agents to the $i$th group, give them an empty bundle, and remove them from consideration. Otherwise, denote by $g$ the last good added to the bundle. Assign to the $i$th group all agents who already have enough utility before $g$ is added, along with an arbitrary subset of those who have enough utility after $g$ is added so that the group has size $n_i$. Allocate the current bundle to the group, and remove the agents and goods involved from consideration. After we have formed $k-1$ groups, simply assign the remaining agents to the $k$th group and give them the leftover goods.

Consider any agent $a_j$. It suffices to show that if at most $k-1$ groups have been formed, the agent still has utility at least $w_j$ for the remaining goods. The statement holds trivially if $w_j\leq 0$, so we may assume that $w_j>0$. If a group takes an empty bundle, $a_j$ loses utility 0. Else, denote by $g$ the last good added to the bundle that a group takes. The agent $a_j$ has utility less than $w_j$ for the bundle before $g$ is added, so she has utility less than $w_j+u_{j,\text{max}}$ for the bundle with $g$. Hence the bundles that have been already allocated to groups are together worth at most $(k-1)(w_j+u_{j,\text{max}}) = \frac{k-1}{k}\cdot u_j(G) + \frac{k-1}{k}\cdot u_{j,\text{max}}$ to $a_j$. This means that $a_j$ has utility at least $u_j(G)-\left(\frac{k-1}{k}\cdot u_j(G) + \frac{k-1}{k}\cdot u_{j,\text{max}}\right) = w_j$ for the remaining goods, as claimed.
\end{proof}

Theorem~\ref{thm:proportional-many-groups} generalizes a result of \citet{Suksompong17}, which holds for individual fair division (i.e., $n_i=1$ for all $i=1,\dots,k$). The factor $\frac{k-1}{k}$ in the approximation cannot be improved even in this special case.

\section{Conclusion and Future Work}

In this paper, we examine the fairness guarantees that can be obtained in the allocation of indivisible goods among groups of agents. For two fixed groups of agents, we provide a complete picture for EF1 and EFX when agents have binary valuations, and we present further positive and negative results for more general valuations. We also introduce a new model where the partition of the agents into groups can be determined along with the allocation of the goods, and show that it is possible to attain a balance in both the agents and the goods simultaneously.

Our work leaves many open questions for future study. For two groups, one could try to establish the existence of EF1 allocations for larger group sizes, either by settling the graph-theoretic conjecture of \citet{JafariAl17} or via other means. 
Another interesting direction is to show existence guarantees for EF1 allocations that simultaneously satisfy some efficiency criterion such as Pareto optimality.
In addition, the questions that we study in this paper can also be asked for multiple groups; our techniques do not seem to extend easily to more than two groups in most cases. 
Specifically, when the goods are divisible, \citet{SegalhaleviSu20} recently showed that an envy-free allocation always exists for an arbitrary number of variable groups of any given sizes---one could try to round such an allocation in order to achieve EF1.
For individual fair division, we also leave the question of whether a balanced EF1 allocation can always be found for any number of agents; Corollary~\ref{cor:EF1-individual-balanced} gives a positive answer for the case of two agents. While the round-robin algorithm works when valuations are responsive, the question intriguingly remains open for arbitrary monotonic valuations.

\subsection*{Acknowledgments}

This work has been partially supported by the European Research Council (ERC) under grant number 639945 (ACCORD).
We would like to thank the anonymous reviewers for their helpful comments.

\bibliographystyle{named}
\bibliography{group-ef}

\newpage

\end{document}